\documentclass[sn-chicago]{sn-jnl}

\usepackage{graphicx}%
\usepackage{multirow}%
\usepackage{amsmath,amssymb,amsfonts}%
\usepackage{amsthm}%
\usepackage{mathrsfs}%
\usepackage[title]{appendix}%
\usepackage{xcolor}%
\usepackage{textcomp}%
\usepackage{manyfoot}%
\usepackage{booktabs}%
\usepackage{algorithm}%
\usepackage{algorithmicx}%
\usepackage{algpseudocode}%
\usepackage{listings}%
\usepackage{float}
\usepackage{tabularx}
\usepackage{booktabs}
\usepackage{geometry}
\theoremstyle{thmstyleone}%
\newtheorem{theorem}{Theorem}
\theoremstyle{thmstyletwo}%
\newtheorem{example}{Example}%
\newtheorem{remark}{Remark}%

\theoremstyle{thmstylethree}%

\begin{document}

\title[Article Title]{
Analysis of hypothesis tests for multiple uncertain finite populations with applications to normal uncertainty distributions}


\author[]{\fnm{Fan} \sur{Zhang}}\email{107552303567@stu.xju.edu.cn}

\author*[]{\fnm{Zhiming} \sur{Li}}\email{zmli@xju.edu.cn}

\affil[]{\orgdiv{School of Mathematics and System Sciences}, \orgname{Xinjiang University}, \orgaddress{\city{Urumqi}, \postcode{830046}, \country{China}}}

\abstract{Hypothesis test plays a key role in uncertain statistics based on uncertain measure. This paper extends the parametric hypothesis of a single uncertain population to multiple cases, thereby addressing a broader range of scenarios. 
First, an uncertain family-wise error rate is defined to control the overall error in simultaneous testing. 
Subsequently, a hypothesis test of two uncertain populations is proposed, and the rejection region for the null hypothesis at a significance level is derived, laying the foundation for further analysis. Building on this, a homogeneity test for multiple populations is developed to assess whether the unknown population parameters differ significantly. When there is no significant difference in these parameters among finite populations or within a subset, a common test is used to determine whether they equal a fixed constant. 
Finally, homogeneity and common tests for normal uncertain populations with means and standard deviations are conducted under three cases: only means, only standard deviations, or both are unknown. Numerical simulations demonstrate the feasibility and accuracy of the proposed methods, and a real example is provided to illustrate their effectiveness. }




\maketitle

\section{Introduction}\label{sec1}
Many observations often exhibit inherent uncertainty rather than randomness, making statistical methods based on probability measures inapplicable. For example, when the number of observations is limited or the
measurement process involves subjectivity, the randomness assumptions underlying traditional statistical methods often become challenging
to verify. Numerous studies have pointed out some evident limitations of the probability theory in addressing real-world problems such as carbon dioxide data (\citet{LiuLiu2025}), fluctuations in financial market interest rates (\citet{YangKe2023}), crude oil spot prices (\citet{ZhangGao2024a}), and queuing systems (\citet{Liuetal2024}). These studies indicate that real-world data are often imprecise or violate core assumptions in probability theory, such as frequency stability, independence, and randomness. To address the types of observations, Liu introduced uncertainty theory, which was founded in 2007 and refined in 2009. It studies belief degrees based on normality, duality, subadditivity, and product axioms. As an essential application of uncertainty theory, \cite{Liu2010} first studied uncertain statistics for collecting, analyzing, and interpreting uncertain data in practice (\cite{WangGaoGuo2012, Liu2024}). For more details on this topic, refer to \cite{Liu2023b}. 

To test whether some statements are correct, hypotheses based on uncertain measures have received considerable attention and are called uncertain hypothesis tests for short. \cite{YeLiu2022} first proposed and applied an uncertain hypothesis test to uncertain regression analysis. Let \(\xi\) be a population with uncertainty distribution \(\Phi(z;\theta)\) for any real value \(z\), where \(\theta\) is an unknown parameter. Given a value \(\theta_0\), a hypothesis test of \(\theta\) is conducted by two statements: \(H_0: \theta=\theta_0\) versus \(H_1:\theta\neq \theta_0\). Here, \(H_0\) is called a null hypothesis, and \(H_1\) is an alternative hypothesis. The key to solving the problem is to choose a suitable region to reject the null hypothesis \(H_0\). 
It is noteworthy that traditional probabilistic hypothesis tests rely on test statistics whose validity depends on large-sample approximations and the assumption of random, precise data. Such reliance can lead to low statistical power in small-sample scenarios, or more fundamentally, when data are imprecise or subjective. One advantage of uncertain hypothesis tests is that they can directly utilize raw data via the inverse uncertainty distribution, offering a framework grounded in uncertainty theory. This approach results in a stricter, more conservative testing procedure and thus provides a distinct theoretical perspective. The application of hypothesis tests under uncertainty in these advanced models often relies on specific distributions. Among these, the normal uncertainty distribution plays a vital role in statistical modeling, particularly for describing disturbance terms.  Due to excellent expansibility, uncertain hypothesis tests have become more and more popular in uncertain regression (\cite{YeLiu2023, JiangYe2023, ZhangLi2025, XiaLi2025}), uncertain time series models (\cite{LiuLi2024, ShiSheng2024}), and uncertain differential equations (\cite{YaoZhangSheng2023, LiXia2024}). 

Based on the above analysis, most previous work focuses on hypothesis testing for a single normal population with uncertainty. However, we often encounter test problems in complex environments involving multiple finite uncertain populations. Therefore, the main contributions and novelties of this paper are summarized as follows:

(i) The definition of an uncertain family-wise error rate is proposed, which provides a theoretical basis for solving the error control problem in multiple testing under uncertain statistics.

(ii) The homogeneity test and the common test are proposed and systematically extended to the scope of uncertain hypothesis testing from parameters of a single population to those of multiple populations.

(iii) For normal uncertain populations, the homogeneity tests and common tests under various scenarios with unknown parameters are thoroughly analyzed, forming a complete testing procedure. 

(iv) Numerical simulations demonstrate the homogeneity and common tests of normal uncertain populations across various parameter settings. A real example is given as an application.

The rest of this paper is organized as follows. Section \ref{sec2} reviews some concepts and introduces the uncertain family-wise error rate, the homogeneity test, and the common test of uncertain populations. 
Section \ref{sec3} focuses on normal uncertain populations and, under different parameter settings, systematically investigates and presents a complete methodology for both homogeneity and common tests.
Section \ref{sec4} performs numerical simulations to illustrate the proposed methods. A real example is given in Section \ref{sec5}, and a brief conclusion is provided in Section \ref{sec6}. 

\section{Uncertain homogeneity and common tests}\label{sec2}
In this section, we first review the basic concepts and hypothesis testing problems for an uncertain population, as presented in \cite{Liu2023b}. Then, we propose the homogeneity and common tests of uncertain populations. 
\subsection{Related definitions and existing results}
Let \(\mathcal {L}\) be a \(\sigma\)-algebra on a nonempty set \(\Gamma\). A set function \(\mathscr{M}: \mathcal {L}\rightarrow [0, 1]\) is called an uncertain measure if it is supposed to satisfy normality, duality, subadditivity, and product axioms. The triple \((\Gamma, \mathcal {L}, \mathscr{M})\) is called an uncertainty space. Let \(\xi\) be an uncertain population with uncertainty distribution \(\Phi(z;\theta)=\mathscr{M}\{\xi\leq z\}\) for \(z\in \mathbb{R}\), where \(\theta (\in \Theta)\) is an unknown parameter. For \(\alpha\in (0, 1)\), the inverse function \(\Phi^{-1}(\alpha; \theta)\) is called the inverse uncertainty distribution of \(\xi\) which is under the regular condition. As defined in \cite{YeLiu2022}, the following hypotheses are called the two-sided hypotheses
\begin{equation}\label{H0}
  H_0: \theta=\theta_0 \text{ versus } H_1: \theta\neq\theta_0,    
\end{equation}
where \(\theta_0\) is a fixed value. Given a vector of observed data \(( z_1, z_2, \cdots, z_n)\) and a significance level \(\alpha\), choosing a suitable rejection region for the null hypothesis \(H_0\) is a vital problem. 

The concept of nonembedded is introduced to ensure that the construction of the rejection region satisfies the optimality conditions of statistical tests. A regular uncertainty distribution family \(\{\Phi(z;\theta) : \theta \in \Theta\}\) is said to be nonembedded for \(\theta_0 \in \Theta\) at level \(\alpha\) if
\[
\Phi^{-1}(\beta;\theta_0) > \Phi^{-1}(\beta;\theta) \quad \text{or} \quad \Phi^{-1}(1-\beta;\theta) > \Phi^{-1}(1-\beta;\theta_0)
\]
for some \(\theta \in \Theta\) and some \(\beta\) with \(0 < \beta \leq \alpha/2\). 

\begin{theorem}\label{thm1}\textit{(\cite{YeLiu2022})}
 Let \(\xi\) be a population with regular uncertainty distribution \(\Phi(z; \theta)\), where \(\theta (\in \Theta)\) is an unknown parameter. If the uncertainty distribution family \(\left\{\Phi(z; \theta) : \theta \in \Theta\right\}\) is nonembedded for a known parameter \(\theta_{0} \in \Theta\) at significance level \(\alpha\), then the test for the hypotheses (\ref{H0}) is
\begin{equation*}
W = \left\{
\begin{aligned}
&\quad~~ (z_{1}, z_{2}, \dots, z_{n}) : \text{ there are at least } \alpha \text{ of indexes } p\text{'s} \text{ with } \\
& 1 \leq p \leq n \text{ such that } z_{p} < \Phi^{-1}(\alpha/2; \theta_{0}) \text{ or } z_{p} > \Phi^{-1}(1 - \alpha/2; \theta_{0})
\end{aligned}
\right\}.
\end{equation*}
\end{theorem}

Let \(\xi\) be a population that follows a normal uncertainty distribution with unknown mean \(e\) and variance \(\sigma^{2}\), denoted by \(\mathscr{N}(e, \sigma)\). Its uncertainty distribution and the inverse uncertainty distribution are  defined by
\begin{align*}
&\Phi(z;e, \sigma) = \left(1 + \exp\left(\frac{\pi(e - z)}{\sqrt{3}\sigma}\right)\right)^{-1}, \quad z \in \mathbb{R},
\\
&\Phi^{-1}(\alpha;e, \sigma)=e+\frac{\sqrt{3}\sigma}{\pi}\ln\left(\frac{\alpha}{1-\alpha}\right), \quad \alpha\in (0, 1). 
\end{align*}
Considering two unknown parameters \(e\) and \(\sigma\), \cite{YeLiu2022} provided the two-sided hypotheses of a normal uncertain population as follows
\begin{equation}
  \label{e sigma }
H_0: e=e_0 \text{ and } \sigma=\sigma_0 \text{ versus } H_1: e \neq e_0 \text{ or } \sigma \neq \sigma_0.  
\end{equation}
Based on Theorem \ref{thm1}, a rejection region for the hypotheses (\ref{e sigma }) at significance level \(\alpha\) is a set 
\begin{equation}
\label{W}
W = \left\{
\begin{aligned}
& (z_{1}, z_{2}, \dots, z_{n}) : \text{ there are at least } \alpha \text{ of indexes } p \text{'s with } 1 \leq p \leq n \\
&\quad \text{ such that } z_{p} < \Phi^{-1}(\alpha/2; e_0, \sigma_0) \text{ or } z_{p} > \Phi^{-1}(1 - \alpha/2; e_0, \sigma_0)
\end{aligned}
\right\}.
\end{equation}
\subsection{ Hypotheses of multiple uncertain finite populations}
In this subsection, we analyze the hypothesis testing problem for multiple populations with uncertainty. When we perform a finite number of null hypotheses simultaneously, denoted by \(H_{0}^i(i=1, \ldots, q)\), controlling the significance level \(\alpha_i\) for each test \(H_{0}^i\) is insufficient to guarantee that the overall error rate does not exceed a certain standard. Therefore, we introduce an uncertain family-wise error rate (UFWER) to ensure that the overall error rate remains within a predetermined significance level when multiple hypothesis tests are performed. Here, the term ``family" refers to the collection of hypotheses \(\{H_{0}^i, i=1, \ldots, q\}\) considered for joint testing. The UFWER is the belief degree of incorrectly rejecting at least one true null hypothesis when multiple null hypotheses are tested simultaneously. Once the family has been defined, control of the UFWER at the joint level \(\alpha\) requires that 
\[
\mbox{UFWER}\leq \alpha
\]
for all possible constellations of true and false hypotheses.

\begin{theorem}
\label{UFWER}
When multiple null hypotheses are tested, if the significance level for each test is $\alpha$, and the UFWER is set to $\alpha$, then to control UFWER, the significance level for each test can remain $\alpha$, provided that the incorrect rejection events are mutually independent.
\end{theorem}

\begin{proof}
Let \(A_1, A_2, \ldots, A_q\) be mutually independent uncertain events, where \(A_i\) represents the \(i\)th test incorrectly rejecting the true null hypothesis \(H_{0}^i\), and \(\mathscr {M}\{A_i\}=\alpha_i\). The belief degree that at least one test incorrectly rejects a true null hypothesis is \(\mathscr {M}\{\underset{i=1}{\overset{q}{\bigcup}} A_i\}\). Thus,
\[
\mbox{UFWER}=\mathscr {M}\left\{\underset{i=1}{\overset{q}{\bigcup}} A_i\right\}=\underset{i=1}{\overset{q}{\bigvee}} \mathscr{M}\{A_i\}=\underset{i=1}{\overset{q}{\bigvee}}\alpha_i\leq \alpha.
\]
In general, take \(\alpha_1=\cdots=\alpha_q=\alpha\). 
Therefore, controlling UFWER at level $\alpha$ does not require reducing the significance level of individual tests below $\alpha$.
\end{proof}

We first consider the case $q=2$. Let \(\xi_i(i=1, 2)\) be two populations with the same regular uncertainty distribution family \(\Phi(z^{(i)};\theta_i)(i=1, 2)\), where \(\theta_i(\in \Theta_i, i=1, 2)\) are unknown parameters. Let \((z^{(i)}_1, z^{(i)}_2, \ldots, z^{(i)}_{m_i})(i=1, 2)\) be a set of observed data of the uncertain population \(\xi_i\). Then, the two populations testing problem for the unknown parameter \(\theta_i\) can be formulated as testing the hypotheses:
\begin{equation}\label{HH0}
H_{0}:\theta_1=\theta_2~\text{versus}~H_{1}:\theta_1 \neq \theta_2.
\end{equation}

\begin{theorem}\label{thm2}
 For the hypotheses (\ref{HH0}), if \(\{\Phi(z^{(i)};\theta_i):\theta_i\in \Theta_i\}\) are nonembedded for known parameters \(\theta_{10}\) and \(\theta_{20}\), 
 after determining that \(\xi_i(i=1, 2)\) respectively follows \(\Phi(z^{(i)};\theta_{i0})\), 
 the rejection region for \(H_{0}\) at the significance level \(\alpha\) is
\begin{equation}
\label{Wij}
W = \left\{
\begin{aligned}
&(z^{(i)}_1, z^{(i)}_2, \dots, z^{(i)}_{m_i}):\exists~i \neq j ~(i,j\in \{1,2\}) \text{ such that at least } \alpha \text{ of indexes } p \text{'s} \\
&~~ \text{with }  1 \leq p \leq m_i \text{ satisfy }z^{(i)}_p < \Phi^{-1} (\alpha/2; \theta_{j0}) \text{ or } z^{(i)}_p > \Phi^{-1} (1-\alpha/2; \theta_{j0})
\end{aligned}
\right\}.
\end{equation}
\end{theorem}

\begin{proof}
Since the uncertainty distribution family \(\Phi(z^{(i)};\theta_i)(i=1, 2)\) are nonembedded for the known parameters \(\theta_{10}\) and \(\theta_{20}\), we can conduct the hypotheses (\ref{HH0}) according to the following four hypothesis tests:

(i) $H_0^1 : \theta_1 = \theta_{10}~\text{versus}~ H_1^1 : \theta_1 \neq \theta_{10}$. Since \(\theta_{10}\) is an estimate calculated using data \((z^{(1)}_1, z^{(1)}_2,\ldots, z^{(1)}_{m_1})\) of \(\xi_1\), the test for \(H_{0}^1\) at significance level \(\alpha\) is
\begin{equation*}
W_1 = \left\{
\begin{aligned}
&(z^{(1)}_1, z^{(1)}_2, \dots, z^{(1)}_{m_1}) : \text{there are at least } \alpha\text{ of indexes } p \text{'s with } 1 \leq p \leq m_1 \\
&\quad\quad~~\text{ such that } z^{(1)}_p < \Phi^{-1}(\alpha/2; \theta_{10}) \text{ or } z^{(1)}_p > \Phi^{-1}(1 - \alpha/2; \theta_{10})
\end{aligned}
\right\}.
\end{equation*}

(ii) $H_0^2 : \theta_2 = \theta_{20}~\text{versus}~ H_1^2 : \theta_2 \neq \theta_{20}$.  Similarly, \(\theta_{20}\) is an estimate calculated through \((z^{(2)}_1, z^{(2)}_2,\ldots, z^{(2)}_{m_2})\) of \(\xi_2\). Hence, the test for \(H_{0}^2\) at significance level \(\alpha\) is
\begin{equation*}
W_2 = \left\{
\begin{aligned}
&(z^{(2)}_1, z^{(2)}_2, \dots, z^{(2)}_{m_2}) : \text{there are at least } \alpha \text{ of indexes } p \text{'s with } 1 \leq p \leq m_2 \\
&\quad\quad~~ \text{ such that } z^{(2)}_p < \Phi^{-1}(\alpha/2; \theta_{20}) \text{ or } z^{(2)}_p > \Phi^{-1}(1 - \alpha/2; \theta_{20})
\end{aligned}
\right\}.
\end{equation*}

In the case that neither \(H_{0}^1\) nor \(H_{0}^2\) can be rejected, we consider that \(\xi_i(i=1, 2)\) follow \(\Phi(z^{(i)};\theta_{i0})\) respectively. 

(iii) $H_0^3 : \theta_1 = \theta_{20}~\text{versus}~ H_1^3 : \theta_1 \neq \theta_{20}$.  Considering a given estimate \(\theta_{20}\), the test for \(H_{0}^3\) at significance level \(\alpha\) is
\begin{equation*}
W_1^2 = \left\{
\begin{aligned}
&(z^{(1)}_1, z^{(1)}_2, \dots, z^{(1)}_{m_1}) : \text{there are at least } \alpha\text{ of indexes } p \text{'s with } 1 \leq p \leq m_1 \\&\quad\quad~~\text{ such that } z^{(1)}_p < \Phi^{-1}(\alpha/2; \theta_{20}) \text{ or } z^{(1)}_p > \Phi^{-1}(1 - \alpha/2; \theta_{20})
\end{aligned}
\right\}.
\end{equation*}

(iv) $H_0^4 : \theta_2 = \theta_{10}~\text{versus}~H_1^4 : \theta_2 \neq \theta_{10}$.  For a given estimate \(\theta_{10}\), the test for \(H_{0}^4\) at significance level \(\alpha\) is
\begin{equation*}
W_2^1 = \left\{
\begin{aligned}
&(z^{(2)}_1, z^{(2)}_2, \dots, z^{(2)}_{m_2}) : \text{there are at least } \alpha \text{ of indexes } p \text{'s with } 1 \leq p \leq m_2 \\&\quad\quad~~ \text{ such that } z^{(2)}_p < \Phi^{-1}(\alpha/2; \theta_{10}) \text{ or } z^{(2)}_p > \Phi^{-1}(1 - \alpha/2; \theta_{10})
\end{aligned}
\right\}.
\end{equation*}

If neither $H_0^1$ nor $H_0^3$ is rejected at significance level $\alpha$, then $\theta_1 = \theta_{10}$ and $\theta_1 = \theta_{20}$, which implies $\theta_1 =\theta_{10} = \theta_{20}$. Similarly, if neither $H_0^2$ nor $H_0^4$ is rejected, then $\theta_2 = \theta_{20}$ and $\theta_2 = \theta_{10}$, implying $\theta_2 =\theta_{20} = \theta_{10}$. Combining these results, by the transitivity of equality, we conclude that $\theta_1 = \theta_2$. 
According to Theorem \ref{UFWER}, the UFWER of the four hypothesis tests is controlled at level $\alpha$. Therefore, we accept the null hypothesis $H_0: \theta_1 = \theta_2$ when none of them are rejected. 
Conversely, if either $H_0^3$ or $H_0^4$ is rejected, we reject $H_0$. The rejection region $W$ is obtained by combining $W_1^2$ and $W_2^1$, yielding the rejection region (\ref{Wij}).
\end{proof}

The acceptance interval for the parameter \(\theta_{j0}\) for the data set \((z^{(i)}_1, z^{(i)}_2, \dots, z^{(i)}_{m_i})\) is denoted as
\[
AI(z^{(i)}; \theta_{j0}) = \left[ \Phi^{-1} \left( \alpha/2; \theta_{j0} \right), \Phi^{-1} \left( 1-\alpha/2; \theta_{j0} \right) \right].
\]
If any \( z^{(i)} \) has more than \(\alpha \times m_i\) data outside its corresponding \( AI(z^{(i)}; \theta_{j0}) \), the null hypothesis \(H_0\) is rejected. 

Theorem \ref{thm2} shows the rejection region of the hypotheses (\ref{HH0}) for two uncertain populations. In practice, we often encounter the problem of testing for homogeneity across multiple populations with uncertainty. Thus, we extend Theorem \ref{thm2} to multiple cases. 

Let \(\xi_i(i=1,2,\ldots,n)\) be the populations with the same regular uncertainty distributions family \(\Phi(z^{(i)};\theta_i)(i=1, 2,\ldots, n)\), respectively, where \(\theta_i(\in \Theta_i, i=1, 2,\ldots, n)\) are unknown parameters. Let \((z^{(i)}_1, z^{(i)}_2,\ldots, z^{(i)}_{m_i})(i=1, 2,\ldots, n)\) be a set of observed data of the uncertain population \(\xi_i\). Thus, the uncertain homogeneity test for the unknown parameter \(\theta_i \) is defined as
\begin{equation}\label{MH0}
\begin{aligned}
H_0 : \theta_1 = \theta_2 = \cdots = \theta_n \ \text{versus}~
H_1 : \theta_1, \theta_2, \ldots, \theta_n \ \text{are not all equal}.
\end{aligned}
\end{equation}
\begin{theorem}\label{thm3}
For the hypotheses (\ref{MH0}), if \(\{\Phi(z^{(i)};\theta_i): \theta_i\in \Theta_i\}\) are nonembedded for known parameters \(\theta_{10}, \theta_{20},..., \theta_{n0}\), 
and \(\xi_i(i=1, 2,..., n)\) respectively follow \(\Phi(z^{(i)};\theta_{i0})\),
then the rejection region for \(H_0\) at the significance level \(\alpha\) is
 \begin{equation} 
 \label{mw}
W = \left\{
\begin{aligned}
&(z^{(i)}_1, z^{(i)}_2, \dots, z^{(i)}_{m_i}):\exists~i \neq j~(i,j\in\{1,2,...,n\})\text{ such that at least } \alpha \text{ of indexes}\\
&~ p \text{'s}  \text{ with } 1 \leq p \leq m_i \text{ satisfy }z^{(i)}_p < \Phi^{-1} (\alpha/2; \theta_{j0}) \text{ or } z^{(i)}_p > \Phi^{-1} (1-\alpha/2; \theta_{j0})
\end{aligned}
\right\}.
\end{equation}
\end{theorem}

\begin{proof}
Rejecting the null hypothesis $H_0: \theta_1 = \theta_2 = \cdots = \theta_n$ is equivalent to the existence of at least one pair of parameters such that $\theta_i \neq \theta_j$ ($i \neq j$). According to Theorem \ref{thm2}, the rejection region \((\ref{Wij})\) provides the criterion for testing any pairwise hypothesis $H_0^{ij}: \theta_i = \theta_j$. Therefore, the multi-population homogeneity test can be achieved by performing all possible pairwise tests: if any one of these pairwise tests rejects $H_0^{ij}$, then the null hypothesis $H_0$ is rejected. Since each pairwise test is conducted at the significance level $\alpha$, and according to Theorem \ref{UFWER}, the UFWER is controlled at level $\alpha$. Thus, the rejection region \((\ref{mw})\) for the hypothesis (\ref{MH0}) is obtained.
\end{proof}

Theorem \ref{thm3} reveals the rejected region of the homogeneity test of multiple uncertain populations. 
The null hypothesis \(H_0\) may be accepted for all populations or rejected. If rejected, nevertheless, specific subsets of populations may exhibit no significant difference in their unknown parameters.
Under the condition \(\theta_1=\theta_2=\cdots =\theta_k=\theta~(2\leq k\leq n)\), the next problem of interest is whether the unknown parameter \(\theta\) is equal to a fixed constant \(\theta_0\), that is, \(\theta=\theta_0\). Note that the \(i\)th population has \(m_i(i=1,\ldots,k)\) sample sizes. It follows that the total number of observations from all populations is denoted by \(N\), satisfying
\(
N = \sum_{i=1}^k m_i.
\)
Assume there are the adjusted and merged data \((z_{1}, z_{2},..., z_{N})\). Therefore, an uncertain common test for the unknown parameter \(\theta\) is defined as
\begin{equation}\label{CH0}
H_{0}^{*}: \theta = \theta_0 \text{ versus } H_{1}^{*}: \theta \neq \theta_0. 
\end{equation}
Based on Theorem \ref{thm1}, the rejection region at significance level \(\alpha\) is
\begin{equation}\label{comW}
W = \left\{
\begin{aligned}
&(z_{1}, z_{2}, \dots, z_{N}) : \text{there are at least } \alpha \text{ of indexes } p \text{'s with } 1 \leq p \leq N \\
&\quad\quad~~\text{ such that } z_{p} < \Phi^{-1}(\alpha/2; \theta_{0}) \text{ or } z_{p} > \Phi^{-1}(1 - \alpha/2; \theta_{0})
\end{aligned}
\right\}.
\end{equation}
If the observed data 
\((z_{1}, z_{2},\ldots, z_{N}) \in W, \)
then we reject \(H_{0}^{*}\). Otherwise, we accept \(H_{0}^{*}\).

\begin{remark}
\rm To facilitate parameter estimation and common test (\ref{CH0}), the observed data related to the unknown parameter \(\theta_i\) may depend on other known parameters, varying across different populations. The data must be adjusted by standardizing these differing known parameters to the same value to address this issue, making the datasets comparable and simplifying parameter estimation and common tests. After this adjustment, all the adjusted datasets can be merged into a single dataset. 
\end{remark}

\section{Hypotheses of normal uncertain populations}\label{sec3}

Suppose the uncertain populations \(\xi_i\sim \Phi(z; e_i,\sigma_i) (i=1, 2,\ldots, n)\), where mean \(e_i\) and standard deviation \(\sigma_i\) are constant parameters. Let \((z^{(i)}_1, z^{(i)}_2, \ldots, z^{(i)}_{m_i})\) be a set of observed data of \(\xi_i (i=1, 2,\ldots, n)\). In the following sections, we will conduct homogeneity and common tests for the parameters, depending on whether their values are known.

\textit{Homogeneity test of mean and deviation parameters}.  For  normal uncertain populations \(\xi_i(i=1,\ldots,n)\), the homogeneity tests of the parameters  $e_i$ and $\sigma_i$ of interest are conducted according to three cases:

\textit{Case i.} For known \(\sigma_i=\sigma_{i0} (i=1, 2,\ldots, n)\), denote \(\Phi(z; e_i) =\Phi(z; e_i,\sigma_{i0})\). The homogeneity test of mean \(e_i (i=1, 2,\ldots, n)\) is expressed by
\begin{equation}\label{ei}
H_0: e_1=e_2=\cdots=e_n.
\end{equation} 

Based on Theorem \ref{thm3}, we can obtain the rejected regions of (\ref{ei}).

\begin{theorem}\label{c1}
For Case i, 
\(\xi_i(i=1,\ldots,n)\) respectively follow \(\Phi(z^{(i)};e_{i0})\).
Then, the rejection region for \(H_{0}\) at the significance level \(\alpha\) is
\begin{equation*}
\label{MW1}
W = \left\{
\begin{aligned}
&(z^{(i)}_1, z^{(i)}_2, \dots, z^{(i)}_{m_i}):\exists~i \neq j~(i,j\in\{1,2,...,n\})\text{ such that at least } \alpha \text{ of indexes}\\
&~ p \text{'s}\text{ with }1 \leq p \leq m_i \text{ satisfy }z^{(i)}_p < \Phi^{-1} (\alpha/2; e_{j0}) \text{ or } z^{(i)}_p > \Phi^{-1} (1-\alpha/2; e_{j0})
\end{aligned}
\right\},
\end{equation*}
where
\(
\Phi^{-1} (\alpha; e_{j0})=e_{j0}+\frac{\sigma_i\sqrt{3}}{\pi}\ln_{}{(\frac{\alpha}{1-\alpha})} =e_{j0}+\sigma_i\Phi_0^{-1}(\alpha).
\)
\end{theorem}

\begin{proof}
Since normal uncertainty distribution family \(\{ \mathscr{N}(e, \sigma) : e \in \mathbb{R}, \sigma > 0 \}\) is nonembedded for any $\theta_0 = (e_0, \sigma_0)$ at any significance level \(\alpha\), it follows from Theorem \ref{thm3} that the rejection region for hypotheses (\ref{ei}) is \(W\).
\end{proof}

\textit{Case ii.} For known \(e_i=e_{i0} (i=1, 2,\ldots, n)\), denote \(\Phi(z; \sigma_i) =\Phi(z; e_{i0},\sigma_i)\). The homogeneity test of standard deviation \(\sigma_i (i=1, 2,\ldots, n)\) satisfies
\begin{equation}\label{sigmai}
H_0: \sigma_1=\sigma_2=\cdots=\sigma_n.
\end{equation} 

Similar to Theorem \ref{c1}, we directly obtain the rejection regions of Cases ii and iii based on Theorem \ref{thm3}.

\begin{theorem}\label{c2}
For Case ii, 
\(\xi_i(i=1,\ldots,n)\) respectively follow \(\Phi(z^{(i)};\sigma_{i0})\).
Then, the rejection region for \(H_{0}\) at the significance level \(\alpha\) is
\begin{equation*}
\label{MW2}
W = \left\{
\begin{aligned}
&(z^{(i)}_1, z^{(i)}_2, \dots, z^{(i)}_{m_i}):\exists~i \neq j~(i,j\in\{1,2,\ldots,n\})\text{ such that at least } \alpha \text{ of indexes}\\
&~~p \text{'s}\text{ with }1 \leq p \leq m_i \text{ satisfy }z^{(i)}_p < \Phi^{-1} (\alpha/2; \sigma_{j0})\text{ or } z^{(i)}_p > \Phi^{-1} (1-\alpha/2; \sigma_{j0})
\end{aligned}
\right\}, 
\end{equation*}
where
\(
\Phi^{-1} (\alpha; \sigma_{j0})=e_i+\frac{\sigma_{j0}\sqrt{3}}{\pi}\ln_{}{(\frac{\alpha}{1-\alpha})} =e_i+\sigma_{j0}\Phi_0^{-1}(\alpha).
\)
\end{theorem}

\textit{Case iii.} If \(e_i\) and \(\sigma_i (i=1, 2,\ldots, n)\) are unknown, the homogeneity test of \(e_i\) and \(\sigma_i (i=1, 2,\ldots, n)\) is written by
\begin{equation}\label{musigma}
H_0: e_1=e_2=\cdots =e_n~\text{and}~\sigma_1=\sigma_2=\cdots =\sigma_n.
\end{equation}

\begin{theorem}\label{c3}
For Case iii, 
\(\xi_i(i=1,\ldots,n)\) respectively follow \(\Phi(z^{(i)};e_{i0}, \sigma_{i0})\).
Then, the rejection region for \(H_{0}\) at the significance level \(\alpha\) is
 \begin{equation*}
 \label{MW3}
W = \left\{
\begin{aligned}
&(z^{(i)}_1, z^{(i)}_2, \dots, z^{(i)}_{m_i}):\exists~i \neq j~(i,j\in\{1,2,\ldots,n\})\text{ such that at least } \alpha \text{ of indexes }p \text{'s}\\
&\text{ with }1 \leq p \leq m_i \text{ satisfy }z^{(i)}_p < \Phi^{-1} (\alpha/2; e_{j0},\sigma_{j0})\text{ or } z^{(i)}_p > \Phi^{-1} (1-\alpha/2;e_{j0}, \sigma_{j0})
\end{aligned}
\right\},
\end{equation*}
where
$
\Phi^{-1} (\alpha; e_{j0}, \sigma_{j0})=e_{j0}+\frac{\sigma_{j0}\sqrt{3}}{\pi}\ln_{}{(\frac{\alpha}{1-\alpha})} =e_{j0}+\sigma_{j0}\Phi_0^{-1}(\alpha).
$
\end{theorem}

After the homogeneity test of normal uncertain populations,
 we may reject or accept the null hypothesis \(H_0\) under Cases i, ii, or iii. 
Through Theorems \ref{c1}-\ref{c3}, if we accept \( H_0 \), there is no significant difference in the unknown parameters among the normal uncertain populations. In this case, we further analyze whether the parameters are equal to a fixed constant. If \( H_0 \) is rejected, it indicates significant differences in the parameters across populations. However, if some populations show no significant difference, we can still test whether these parameters are equal to a fixed constant.

\textit{Common tests of mean and deviation parameters}.   Suppose the homogeneity test shows among \( n \) populations, there are \( k~(2\leq k\leq n) \) populations (for example, \( \xi_1, \xi_2, \ldots, \xi_k \)) have no significant differences in their unknown parameters, i.e., \( \theta_1 = \theta_2 = \cdots = \theta_k \). Based on the three cases below, we propose common tests for normal uncertain populations.

\textit{Case I.} Under \(e_1=e_2=\cdots=e_k=e\), the common test 
for \( e \) being equal to a fixed constant \( e_0 \) is expressed by 
\begin{equation}\label{e0}
H_{0}^{*}: e = e_0 \text{ versus } H_{1}^{*}: e \neq e_0.
\end{equation} 

Under Case I, \(\xi_i\) has the uncertainty distribution \(\mathscr{N}(e,\sigma_i)(i=1,2,...,k)\). To facilitate parameter estimation, we need to adjust the data \((z^{(i)}_1, z^{(i)}_2, \ldots,z^{(i)}_{m_i})(i=1,2,...,k)\), to unify the standard deviations of different datasets without changing the means, achieved through the transformation:
\[
\frac{z_l^{(i)}-e_i}{\sigma_i}+e_i,\quad l=1,2,...,m_i,i=1,2,...,k.
\]
In this way, each set of adjusted data has the same uncertainty distribution \(\mathscr{N}(e,1)\). 
Then, \(k\) sets of adjusted data are merged into a data set, denoted  by
\begin{align*}
\{z_{1}, z_{2},..., z_{N}\}=& \left\{\frac{z_{1}^{(1)}-e_{1}}{\sigma_1}~+e_{1}, \frac{z_{2}^{(1)}-e_{1}}{\sigma_1}~+e_{1}, \ldots, \frac{z_{m_1}^{(1)}-e_{1}}{\sigma_1}+e_{1},\frac{z_{1}^{(2)}-e_{2}}{\sigma_2}~+e_{2}, \frac{z_{2}^{(2)}-e_{2}}{\sigma_2}~+e_{2}, \ldots,\right. \\
&~~~\frac{z_{m_2}^{(2)}-e_{2}}{\sigma_2}+e_{2}, \ldots, \left. \frac{z_{1}^{(k)}-e_{k}}{\sigma_k}+e_{k}, \frac{z_{2}^{(k)}-e_{k}}{\sigma_k}+e_{k}, \ldots, \frac{z_{m_k}^{(k)}-e_{k}}{\sigma_k}+e_{k}\right\}
\end{align*}
for parameter estimation and the common test (\ref{e0}). From (\ref{comW}), the rejection region for \(H_{0}^*\) in (\ref{e0}) at significance level \(\alpha\) is
\begin{equation*}
W = \left\{
\begin{aligned}
&(z_{1}, z_{2}, \dots, z_{N}) : \text{there are at least } \alpha \text{ of indexes } p \text{'s with } 1 \leq p \leq N \\
&\quad\quad~~~\text{such that }z_{p} < \Phi^{-1}(\alpha/2; e_{0}) \text{ or } z_{p} > \Phi^{-1}(1 - \alpha/2; e_{0})
\end{aligned}
\right\},
\end{equation*}
where 
$ \Phi^{-1} (\alpha; e_0)=e_0+\frac{\sqrt{3}}{\pi}\ln_{}{(\frac{\alpha}{1-\alpha})}.$

\textit{Case II.} Under \(\sigma_1=\sigma_2=\cdots=\sigma_k=\sigma\), the common test 
for \( \sigma \) is represented as
\begin{equation}\label{sigma0}
H_{0}^{*}: \sigma = \sigma_0 \text{ versus } H_{1}^{*}: \sigma \neq \sigma_0.
\end{equation} 

For Case II, \(\xi_i\) has the uncertainty distribution \(\mathscr{N}(e_i,\sigma)(i=1,2,\ldots,k)\). Similar to Case I, we need to adjust the data \((z^{(i)}_1, z^{(i)}_2, \ldots, z^{(i)}_{m_i})(i=1,2,\ldots,k)\) to unify the means of different datasets without changing the standard deviations, i.e., \(z_l^{(i)}-e_i(l=1,2,\ldots,m_i, i=1,2,\ldots,k)\). Thus, each set of adjusted data has the same uncertainty distribution \(\mathscr{N}(0,\sigma)\). Further, we adjust the \(k\) sets of data and combine them into one set as follows:
\begin{align*}
\{z_{1}, z_{2},..., z_{N}\}=&\{z_{1}^{(1)}-e_1, z_{2}^{(1)}-e_1,..., z_{m_1}^{(1)}-e_1, z_{1}^{(2)}-e_2, z_{2}^{(2)}-e_2,...,\\
&~ z_{m_2}^{(2)}-e_2,..., z_{1}^{(k)}-e_k, z_{2}^{(k)}-e_k,..., z_{m_k}^{(k)}-e_k\}.
\end{align*}
Given the data \(\{z_{1}, z_{2},..., z_{N}\}\), through (\ref{comW}), the rejection region of \(H_{0}^*\) in (\ref{sigma0}) at significance level \(\alpha\) is
\begin{equation*}
W = \left\{
\begin{aligned}
& (z_{1}, z_{2}, \dots, z_{N}) : \text{there are at least } \alpha \text{ of indexes } p \text{'s with } 1 \leq p \leq N \\
&\quad\quad~~~ \text{such that }z_{p} < \Phi^{-1}(\alpha/2; \sigma_{0}) \text{ or } z_{p} > \Phi^{-1}(1 - \alpha/2; \sigma_{0})
\end{aligned}
\right\},
\end{equation*}
where 
$ \Phi^{-1} (\alpha;\sigma_0)=\sigma_0\frac{\sqrt{3}}{\pi}\ln_{}{(\frac{\alpha}{1-\alpha})}.$

\textit{Case III.} Under \(e_1=e_2=\cdots=e_k=e\) and \(\sigma_1=\sigma_2=\cdots=\sigma_k=\sigma\), the common test 
for \( e \) and \(\sigma\) is shown as
\begin{equation}\label{musigma0}
\begin{aligned}
H_{0}^{*}: e = e_0 \text{ and } \sigma = \sigma_0 \text{ versus } H_{1}^{*}: e \neq e_0 \text{ or } \sigma \neq \sigma_0.
\end{aligned}
\end{equation}

Under Case III, it is obvious that \(\xi_i(i=1,\ldots,k)\) has the uncertainty distribution \(\mathscr{N}(e,\sigma)\). Hence, the \(k\) sets of data \((z^{(i)}_1, z^{(i)}_2, \ldots, z^{(i)}_{m_i})(i=1,2,\ldots,k)\) can be merged into one set without adjustment, as shown below.
\begin{align*}
\{z_{1}, z_{2},..., z_{N}\}=\{z_{1}^{(1)}, z_{2}^{(1)},..., z_{m_1}^{(1)}, z_{1}^{(2)}, z_{2}^{(2)},..., z_{m_2}^{(2)},..., z_{1}^{(k)}, z_{2}^{(k)},..., z_{m_k}^{(k)}\}.
\end{align*}
Given a significance level \(\alpha\), the rejection region of \(H_{0}^*\) in (\ref{musigma0}) is
\begin{equation*}
W = \left\{
\begin{aligned}
&(z_{1}, z_{2}, \dots, z_{N}) : \text{there are at least } \alpha\text{ of indexes } p \text{'s with } 1 \leq p \leq N \\
&\quad \text{ such that }z_{p} < \Phi^{-1}(\alpha/2;e_0, \sigma_{0}) \text{ or } z_{p} > \Phi^{-1}(1 - \alpha/2;e_0, \sigma_{0})
\end{aligned}
\right\},
\end{equation*}
where 
$ \Phi^{-1} (\alpha; e_0, \sigma_0)=e_0+\sigma_0\frac{\sqrt{3}}{\pi}\ln_{}{(\frac{\alpha}{1-\alpha})}. $

\section{Numerical simulations}\label{sec4}
This section presents three numerical simulations of homogeneity and common tests for normal uncertain populations under various parameter settings. 

\begin{example}\label{exam1}
\rm Let \(\xi_i(i=1,2,3)\) be three normal uncertain populations with normal uncertainty distributions \(\mathscr{N}(4.5, \sigma_1)\), \(\mathscr{N}(5, \sigma_2)\), \(\mathscr{N}(5.5, \sigma_3)\), where   \(\sigma_i(i=1,2,3)\) are unknown parameters.  And \((z^{(i)}_1, z^{(i)}_2, \ldots, z^{(i)}_{m_i})\) is the sample of the uncertain populations \(\xi_i(i=1,2,3)\), as shown in Table \ref{Data of Example 1}. To test whether \(\sigma_i(i=1,2,3)\) of uncertain populations are equal, we consider the following hypotheses: 
\[
H_0:\sigma_1=\sigma_2=\sigma_3~\text{versus}~H_1:\sigma_1, \sigma_2, \sigma_3~\text{are not all equal}.
\]
Utilizing the method of moments, we obtain \(\sigma_{10} =1.420 \), \(\sigma_{20} = 1.348\), \(\sigma_{30} = 1.434\). 
 The hypothesis tests confirm that \(\xi_i(i=1,2,3)\) follow \(\mathscr{N}(4.5, 1.420)\), \(\mathscr{N}(5, 1.348)\), and \(\mathscr{N}(5.5, 1.434)\), respectively.
 It follows from Theorem \ref{c2} that the rejection region of \(H_{0}\) at the significance level \(\alpha=0.05\) is
 \begin{equation*}
W = \left\{
\begin{aligned}
&(z^{(i)}_1, z^{(i)}_2, \dots, z^{(i)}_{m_i}):\exists~i \neq j~(i,j\in\{1,2,3\})\text{ such that at least } \alpha \text{ of indexes } p\text{'s}\\
&~~\text{ with }1 \leq p \leq m_i \text{ satisfy }z^{(i)}_p < \Phi^{-1} (\alpha/2; \sigma_{j0})\text{ or } z^{(i)}_p > \Phi^{-1} (1-\alpha/2; \sigma_{j0})
\end{aligned}
\right\}, 
\end{equation*}
where
$
\Phi^{-1} (\alpha; \sigma_{j0})=e_i+\frac{\sigma_{j0}\sqrt{3}}{\pi}\ln_{}{(\frac{\alpha}{1-\alpha})} =e_i+\sigma_{j0}\Phi_0^{-1}(\alpha).
$

\begin{table}[h]
 \centering
 \caption{Samples of the populations \(\xi_i (i=1,2,3)\) in Example \ref{exam1}.}
 \label{Data of Example 1}
 \begin{tabular*}{\hsize}{@{}@{\extracolsep{\fill}}cccccccccccccc@{}}
 \hline
 \(i\) & \(j\) & \multicolumn{12}{c}{\(z^{(i)}_j\)} \\ \hline
 \multirow{3}{*}{1} 
 & 1-12  & 6.38 & 4.56 & 0.01 & 5.25 & 4.16 & 4.59 & 5.81 & 4.61 & 4.64 & 6.92 & 6.70 & 5.05 \\
 & 13-24 & 2.94 & 3.60 & 2.46 & 3.29 & 6.37 & 3.29 & 2.51 & 3.34 & 5.22 & 5.57 & 4.82 & 3.96 \\
 & 25-36 & 3.27 & 3.83 & 4.96 & 4.56 & 3.57 & 5.33 & 4.70 & 4.18 & 6.43 & 6.57 & 6.07 & 4.36 \\
 \hline
 \multirow{4}{*}{2} 
 & 1-12  & 4.16 & 4.93 & 4.91 & 2.51 & 5.21 & 5.29 & 8.45 & 6.15 & 5.41 & 2.85 & 3.80 & 6.12 \\
 & 13-24 & 4.12 & 6.68 & 3.52 & 5.58 & 3.42 & 4.95 & 5.17 & 5.82 & 5.28 & 5.70 & 5.46 & 4.97 \\
 & 25-36 & 4.39 & 4.61 & 2.53 & 5.80 & 4.17 & 3.64 & 5.27 & 3.34 & 1.72 & 7.40 & 6.12 & 5.16 \\
 & 37-48 & 5.40 & 7.59 & 5.85 & 6.09 & 5.87 & 4.23 & 5.83 & 4.68 & 3.56 & 2.91 & 4.36 & 6.15 \\
 \hline
 \multirow{5}{*}{3}
 & 1-12  & 6.29 & 4.49 & 5.68 & 5.29 & 5.91 & 2.92 & 7.51 & 4.26 & 5.59 & 5.34 & 8.88 & 5.74 \\
 & 13-24 & 1.36 & 5.30 & 5.10 & 6.84 & 5.13 & 7.78 & 4.21 & 4.13 & 2.85 & 6.00 & 4.65 & 5.56 \\
 & 25-36 & 5.50 & 5.46 & 3.96 & 5.86 & 5.07 & 6.19 & 6.39 & 1.24 & 3.07 & 5.65 & 4.30 & 6.08 \\
 & 37-48 & 4.28 & 7.26 & 6.63 & 4.93 & 4.09 & 4.28 & 5.92 & 7.14 & 5.33 & 5.37 & 6.46 & 7.01 \\
 & 49-60 & 5.86 & 7.98 & 4.28 & 5.70 & 5.96 & 7.83 & 4.55 & 5.68 & 6.62 & 5.15 & 6.33 & 5.19 \\
 \hline
 \end{tabular*}
\end{table}

The results of the uncertain homogeneity test are shown in Table \ref{results of the uncertain homogeneity test in Example 1} and Fig. \ref{results of the uncertain homogeneity test in Example 1 (f)}. Since \( 36 \times 0.05=1.8, 48 \times 0.05=2.4, 60 \times 0.05=3\), we can see that there only 1 \(z_j^{(1)}(j=3)\) not in \(AI(z^{(1)};\sigma_{i0})(i=2,3)\), only 2 \(z_j^{(2)}(j=7,33)\) not in \(AI(z^{(2)};\sigma_{i0})(i=1,3)\), only 3 \(z_j^{(3)}(j=11,13,32)\) not in \(AI(z^{(3)};\sigma_{i0})(i=1,2)\); that is to say, \(H_0\)  cannot be rejected, thus there is no significant difference in the standard deviations of \(\xi_1, \xi_2,\xi_{3} \) at the significance level 0.5. 


\begin{table}[h]
\centering
\caption{Results of the uncertain homogeneity test in Example \ref{exam1}.}
\label{results of the uncertain homogeneity test in Example 1}
\begin{tabular*}{\hsize}{@{}@{\extracolsep{\fill}}cccc@{}}
\hline
\hspace{1em}\(z^{(i)}\) &\hspace{1em} \(AI(z^{(i)};\sigma_{10})\) &\hspace{1em} \(AI(z^{(i)};\sigma_{20})\) &\hspace{1em} \(AI(z^{(i)};\sigma_{30})\)\hspace{1em} \\ 
\hline
\hspace{1em}\(z^{(1)}\) &\hspace{1em} [1.632, 7.368] &\hspace{1em} [1.778, 7.223] &\hspace{1em} [1.603, 7.397]\hspace{1em} \\ 
\hspace{1em}\(z^{(2)}\) &\hspace{1em} [2.132, 7.868] &\hspace{1em} [2.278, 7.723] &\hspace{1em} [2.103, 7.897]\hspace{1em} \\ 
\hspace{1em}\(z^{(3)}\) &\hspace{1em} [2.632, 8.368] &\hspace{1em} [2.778, 8.223] &\hspace{1em} [2.603, 8.397]\hspace{1em} \\ 
\hline
\end{tabular*}
\end{table}

\begin{figure}[h]
\centering
\includegraphics[width=6in]{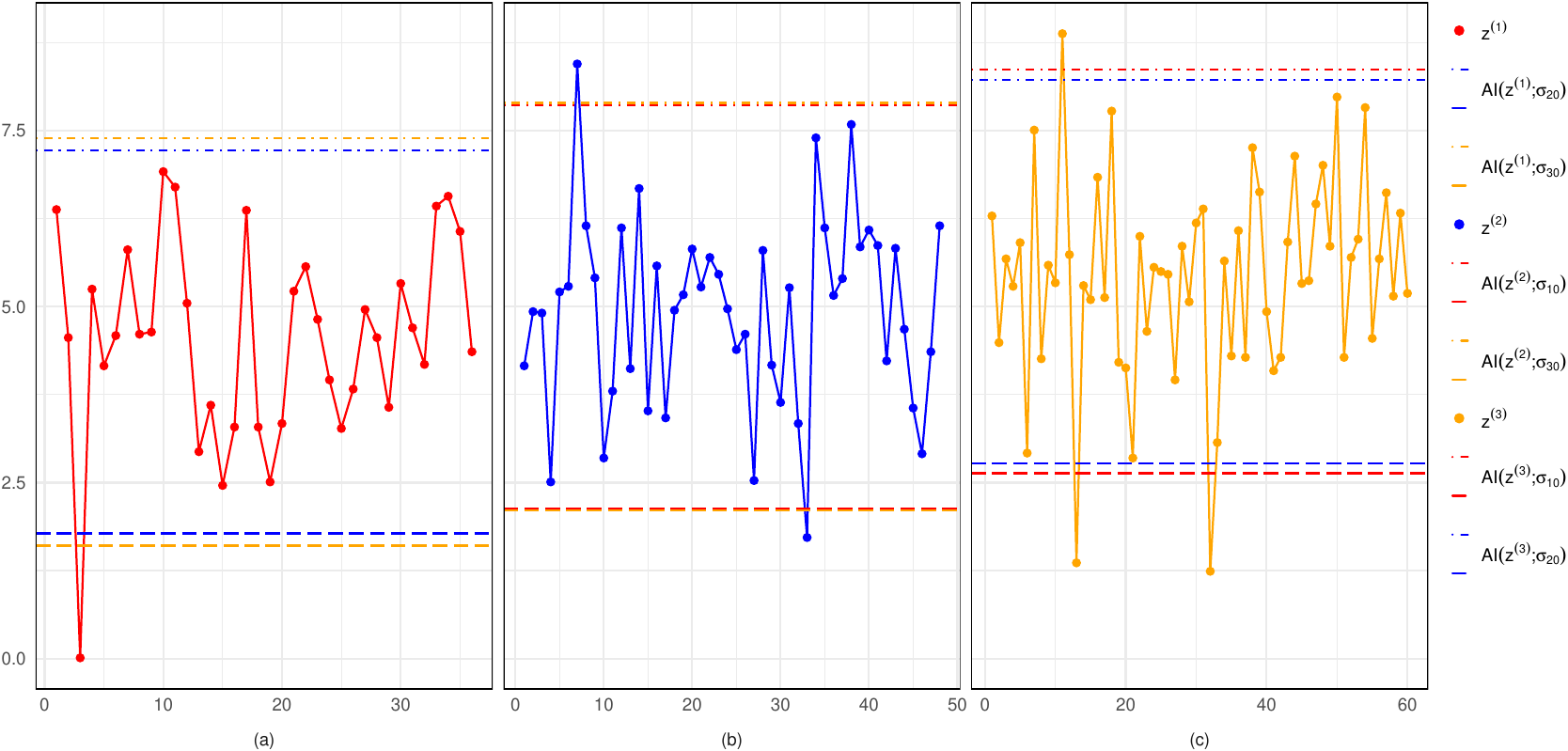}
\caption{Results of the uncertain homogeneity test in Example 1.
}
\label{results of the uncertain homogeneity test in Example 1 (f)}
\end{figure}

Since there is no significant difference in the unknown parameters of each population, the uncertain common test can be performed to determine whether the unknown parameter \(\sigma_1=\sigma_2=\sigma_3=\sigma\) is equal to a fixed constant \(\sigma_0\). For this, we consider the following hypotheses:
\[
H_{0}^{*}: \sigma = \sigma_0 \text{ versus } H_{1}^{*}: \sigma \neq \sigma_0.
\]
3 adjusted sets of data are then combined into one set, i.e.,
\begin{align*}
\{z_{1}, z_{2},..., z_{144}\}=&\{z_{1}^{(1)}-4.5, z_{2}^{(1)}-4.5,..., z_{36}^{(1)}-4.5, z_{1}^{(2)}-5, z_{2}^{(2)}-5,..., z_{48}^{(2)}-5, \\&~z_{1}^{(3)}-5.5, z_{2}^{(3)}-5.5,..., z_{60}^{(3)}-5.5\},
\end{align*}
we can calculate
$
\sigma_0^2 = \frac{1}{144} \sum_{i=1}^{144} (z_i)^2=1.404^2, 
$
given a significance level \(\alpha=0.05\), we obtain
$
\Phi^{-1}(\alpha/2;\sigma_0)=-2.836, \Phi^{-1}(1-\alpha/2;\sigma_0)=2.836,$
where
$
\Phi^{-1}(\alpha;\sigma_0)=\frac{1.404\sqrt{3}}{\pi}\ln_{}{(\frac{\alpha}{1-\alpha})}.
$
Since \(\alpha \times 144 =7.2\), we have
\begin{equation*}
W = \left\{
\begin{aligned}
& (z_{1}, z_{2}, \dots, z_{144}) : \text{there are at least } 8\text{ of indexes } p \text{'s with } \\
&\quad~~ 1 \leq p \leq 144 \text{ such that } z_p < -2.836 \text{ or } z_p > 2.836
\end{aligned}
\right\}.
\end{equation*}
Only 6 \(z_j(j=3, 43, 69, 95, 97, 116)\) not in \([-2.836, 2.836]\). Thus \((z_{1}, z_{2},..., z_{144}) \notin W\). \(H_{0}^{*}\) cannot be rejected, that is to say the unknown parameter \(\sigma\) is equal to a fixed constant \(\sigma_0=1.404\).
\end{example}

\begin{example}\label{exam2}
\rm Let \(\xi_i(i=1,2,3)\) be three normal uncertain populations with normal uncertainty distributions \(\mathscr{N}(e_1, 1), \mathscr{N}(e_2, 1.5), \mathscr{N}(e_3, 2)\), where means \(e_i(i=1,2,3)\) are unknown parameters. And
\((z^{(i)}_1, z^{(i)}_2, \ldots, z^{(i)}_{m_i})\) represent the samples of the uncertain populations \(\xi_i(i=1,2,3)\), as presented in Table \ref{Data of Example 2}.  To test whether the means \(e_i(i=1,2,3)\) of uncertain populations are equal, we consider the following hypotheses: 
\[
H_0:e_1=e_2=e_3~\text{versus}~H_1:e_1, e_2, e_3~\text{are not all equal}.
\]
\begin{table}[h]
 \centering
 \caption{Samples of the populations \(\xi_i (i=1,2,3)\) in Example \ref{exam2}.}
 \label{Data of Example 2}
 \begin{tabular*}{\hsize}{@{}@{\extracolsep{\fill}}cccccccccccccc@{}}
 \hline
 \(i\) & \(j\) &\multicolumn{12}{c}{\(z^{(i)}_j\)} \\ \hline
 \multirow{4}{*}{\(1\)} 
 & \(1-12\)  & 6.32 & 3.54 & 4.81 & 4.87 & 2.33 & 5.24 & 5.42 & 4.64 & 5.03 & 4.64 & 5.45 & 4.93 \\
 & \(13-25\) & 4.18 & 4.69 & 4.64 & 7.47 & 3.51 & 3.04 & 5.37 & 5.31 & 4.49 & 4.27 & 6.64 & 4.01 \\
 & \(25-36\) & 5.28 & 4.72 & 5.67 & 5.86 & 4.11 & 4.14 & 4.96 & 5.38 & 5.01 & 5.78 & 3.59 & 5.16 \\
 & \(37-48\) & 4.88 & 4.59 & 4.25 & 5.01 & 5.33 & 4.69 & 6.28 & 6.29 & 6.30 & 3.92 & 5.58 & 5.89 \\
 \hline
 \multirow{3}{*}{\(2\)} 
 & \(1-12\)  & 5.79 & 4.23 & 7.86 & 6.88 & 4.02 & 7.22 & 7.26 & 4.56 & 6.28 & 2.94 & 4.62 & 4.36 \\
 & \(13-24\) & 5.61 & 4.21 & 6.55 & 5.53 & 4.31 & 5.06 & 5.43 & 4.59 & 6.99 & 5.65 & 4.42 & 3.64 \\
 & \(25-36\) & 3.71 & 4.76 & 5.13 & 3.71 & 6.97 & 5.57 & 4.65 & 7.31 & 4.25 & 4.68 & 5.13 & 5.17 \\
 \hline
 \multirow{5}{*}{\(3\)}
 & \(1-12\)  & 6.16 & 4.49 & 2.17 & 2.13 & 4.90 & 4.96 & 5.34 & 7.28 & 8.52 & 5.74 & 4.60 & 3.56 \\
 & \(13-24\) & 4.73 & 2.26 & 4.60 & 6.75 & 7.02 & 1.57 & 4.10 & 7.62 & 8.69 & 6.45 & 5.61 & 4.08 \\
 & \(25-36\) & 2.90 & 5.94 & 8.12 & 1.91 & 5.38 & 7.52 & 4.95 & 4.57 & 6.55 & 5.71 & 1.52 & 5.13 \\
 & \(37-48\) & 6.47 & 4.78 & 6.65 & 4.91 & 3.61 & 8.08 & 7.47 & 5.21 & 1.07 & 5.74 & 3.97 & 4.04 \\
 & \(49-60\) & 3.85 & 4.34 & 4.06 & 5.53 & 4.54 & 2.63 & 4.57 & 5.09 & 5.08 & 5.38 & 7.85 & 6.49 \\
 \hline
 \end{tabular*}
\end{table}
Utilizing the method of moments, we obtain \(e_{10} =4.948 \), \(e_{20} = 5.251\), \(e_{30} = 5.082\). 
 The hypothesis tests confirm that \(\xi_i(i=1,2,3)\) follow \(\mathscr{N}(4.948,1)\), \(\mathscr{N}(5.251,1.5)\), \(\mathscr{N}(5.082,2)\), respectively.
 It follows from Theorem \ref{c1} that the rejection region of \(H_{0}\) at the significance level \(\alpha=0.05\) is 
\begin{equation*}
W = \left\{
\begin{aligned}
&~(z^{(i)}_1, z^{(i)}_2, \dots, z^{(i)}_{m_i}):\exists~i \neq j~(i,j\in\{1,2,3\})\text{ such that at least } \alpha \text{ of indexes}\\
& p \text{'s}\text{ with }1 \leq p \leq m_i \text{ satisfy }z^{(i)}_p < \Phi^{-1} (\alpha/2; e_{j0})\text{ or } z^{(i)}_p > \Phi^{-1} (1-\alpha/2; e_{j0})
\end{aligned}
\right\},
\end{equation*}
where
$
\Phi^{-1} (\alpha; e_{j0})=e_{j0}+\frac{\sigma_i\sqrt{3}}{\pi}\ln_{}{(\frac{\alpha}{1-\alpha})} =e_{j0}+\sigma_i\Phi_0^{-1}(\alpha). 
$

The results of the uncertain homogeneity test are shown in Table \ref{results of the uncertain homogeneity test in Example 2} and Fig. \ref{results of the uncertain homogeneity test in Example 2 (f)}. Since \(48 \times 0.05=2.4, 36 \times 0.05=1.8, 60 \times 0.05=3\), we can see that there are 3 \(z_j^{(1)}(j=5,16,18)\) not in \(AI(z^{(1)};\sigma_{i0})(i=2,3)\). Thus, there is a significant difference in the standard deviations of \(\xi_1\) and \(\xi_i(i=2,3)\) at the significance level 0.5. The results show that \(H_{0}\) should be rejected.

\begin{table}[h]
\centering
\caption{Results of the uncertain homogeneity test in Example \ref{exam2}.}
\label{results of the uncertain homogeneity test in Example 2}
\begin{tabular*}{\hsize}{@{}@{\extracolsep{\fill}}cccc@{}}
\hline
\hspace{1em}\(z^{(i)}\) &\hspace{1em} \(AI(z^{(i)};e_{10})\) &\hspace{1em} \(AI(z^{(i)};e_{20})\) &\hspace{1em} \(AI(z^{(i)};e_{30})\)\hspace{1em} \\ 
\hline
\hspace{1em}\(z^{(1)}\) &\hspace{1em} [2.928, 6.968] &\hspace{1em} [3.232, 7.271] &\hspace{1em} [3.063, 7.102]\hspace{1em} \\ 
\hspace{1em}\(z^{(2)}\) &\hspace{1em} [1.918, 7.978] &\hspace{1em} [2.222, 8.281] &\hspace{1em} [2.053, 8.112]\hspace{1em} \\ 
\hspace{1em}\(z^{(3)}\) &\hspace{1em} [0.909, 8.988] &\hspace{1em} [1.212, 9.291] &\hspace{1em} [1.043, 9.122]\hspace{1em} \\ 
\hline
\end{tabular*}
\end{table}

\begin{figure}[h]
\centering
\includegraphics[width=6in]{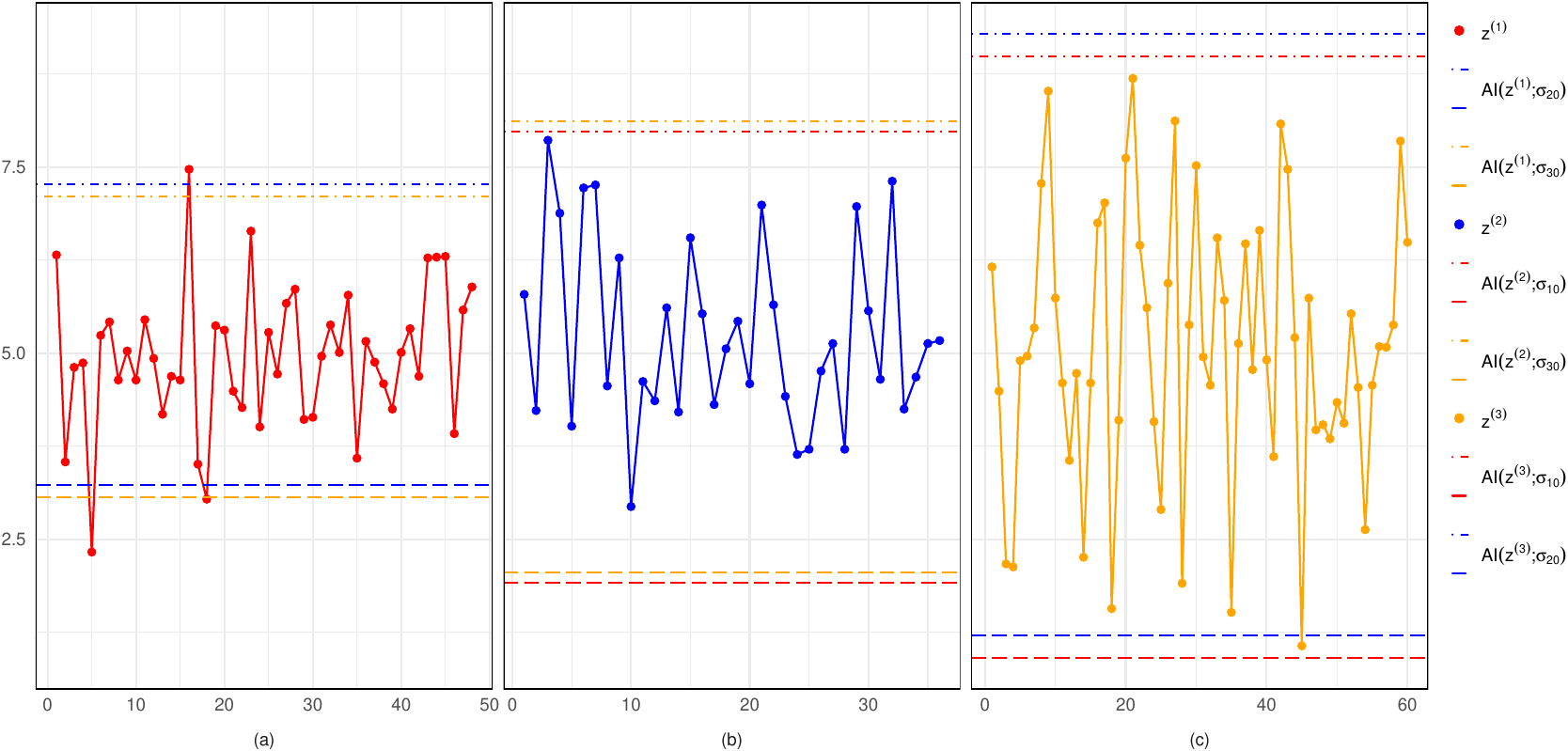}
\caption{Results of the uncertain homogeneity test in Example \ref{exam2}. 
}
\label{results of the uncertain homogeneity test in Example 2 (f)}
\end{figure}

 Since there is no significant difference in the unknown parameters of \(\xi_2\) and \(\xi_3\), the uncertain common test can be performed to determine whether the unknown parameter \(e_2=e_3=e\) is equal to a fixed constant \(e_0\). For this, we consider the following hypotheses:
\[
H_{0}^{*}: e = e_0\text{ versus } H_{1}^{*}: e \neq e_0.
\]
2 adjusted sets of data are then combined into one set, i.e.,
\begin{align*}
\{z_{1}, z_{2},..., z_{144}\}=&
\left\{~\frac{z_{1}^{(2)}-5.251}{1.5}+5.251, \ldots, \frac{z_{36}^{(2)}-5.251}{1.5}+5.251, \right. \\
&\left.\frac{z_{1}^{(3)}-5.082}{2}+5.082, \ldots, \frac{z_{60}^{(3)}-5.082}{2}+5.082~\right\},
\end{align*}
we can calculate
$
e_0 = \frac{1}{96} \sum_{i=1}^{96} z_i=5.146, 
$
given a significance level \(\alpha=0.05\), we obtain
\(
\Phi^{-1}(\alpha/2;e_0)=3.126, \Phi^{-1}(1-\alpha/2;e_0)=7.166,
\)
where 
$
 \Phi^{-1}(\alpha;e_0)= 5.146+\frac{\sqrt{3}}{\pi}\ln_{}{(\frac{\alpha}{1-\alpha})}.
$
 Since \(\alpha \times 96 =4.8\), it follows that
\begin{equation*}
W = \left\{
\begin{aligned}
&~(z_{1}, z_{2}, \dots, z_{96}) : \text{there are at least } 5 \text{ of indexes } p \text{'s}\\
&\text{with }  1 \leq p \leq 96 \text{ such}~\text{that } z_p < 3.126 \text{ or } z_p > 7.166
\end{aligned}
\right\}.
\end{equation*}
\vspace{2pt}
Only \(z_{81}\) not in \([3.126, 7.166]\). Thus, \((z_{1}, z_{2},..., z_{96}) \notin W\). \(H_{0}^{*}\) cannot be rejected; that is, the unknown parameter \(e\) is equal to a fixed constant \(e_0= 5.146\).
\end{example}

\begin{example}\label{exam3}
\rm Let \(\xi_i(i=1,2,3)\) be three normal uncertain populations with normal uncertainty distributions \(\mathscr{N}(e_1, \sigma_1), \mathscr{N}(e_2, \sigma_2), \mathscr{N}(e_3, \sigma_3)\), where means \(e_i(i=1,2,3)\) and standard deviations \(\sigma_i(i=1,2,3)\) are unknown parameters. The samples \( (z^{(i)}_1, z^{(i)}_2, \ldots, z^{(i)}_{m_i}) \) are derived from the uncertain populations \(\xi_i (i = 1, 2, 3)\), illustrated in Table \ref{Data of Example 3}.

\begin{table}[h]
 \centering
 \caption{Samples of the populations \(\xi_i (i=1,2,3)\) in Example \ref{exam3}.}
 \label{Data of Example 3}
 \begin{tabular*}{\hsize}{@{}@{\extracolsep{\fill}}cccccccccccccc@{}}
 \hline
 \(i\) & \(j\) & \multicolumn{12}{c}{\(z^{(i)}_j\)} \\ \hline
 \multirow{3}{*}{1} 
 & 1-12  & 5.79 & 4.23 & 7.86 & 6.88 & 4.02 & 7.22 & 7.26 & 4.56 & 6.28 & 2.94 & 4.62 & 4.36 \\
 & 13-24 & 5.61 & 4.21 & 6.55 & 5.53 & 4.31 & 5.06 & 5.43 & 4.59 & 6.99 & 5.65 & 4.42 & 3.64 \\
 & 25-36 & 3.71 & 4.76 & 5.13 & 3.71 & 6.97 & 5.57 & 4.65 & 7.31 & 4.25 & 4.68 & 5.13 & 5.17 \\
 \hline
 \multirow{4}{*}{2} 
 & 1-12  & 5.03 & 3.99 & 4.26 & 7.12 & 3.50 & 4.23 & 3.60 & 4.58 & 5.71 & 5.07 & 3.91 & 5.42 \\
 & 13-24 & 5.25 & 5.64 & 5.62 & 6.82 & 7.21 & 5.38 & 2.42 & 3.37 & 4.94 & 4.14 & 5.59 & 3.87 \\
 & 25-36 & 5.85 & 2.96 & 7.21 & 3.08 & 5.58 & 4.71 & 4.50 & 4.94 & 3.55 & 4.75 & 4.03 & 3.90 \\
 & 37-48 & 6.00 & 7.27 & 5.61 & 5.20 & 5.82 & 3.34 & 8.42 & 5.35 & 4.66 & 7.11 & 5.29 & 3.34 \\
 \hline
 \multirow{5}{*}{3} 
 & 1-12  & 5.99 & 7.20 & 6.60 & 3.04 & 7.30 & 6.14 & 5.84 & 4.60 & 3.66 & 5.71 & 3.09 & 4.15 \\
 & 13-24 & 6.85 & 4.79 & 3.34 & 4.16 & 5.96 & 5.25 & 5.01 & 5.43 & 2.95 & 5.72 & 3.98 & 5.59 \\
 & 25-36 & 5.79 & 6.98 & 3.87 & 5.43 & 5.42 & 4.99 & 5.62 & 5.31 & 4.77 & 6.21 & 3.61 & 3.41 \\
 & 37-48 & 7.72 & 5.66 & 4.39 & 5.94 & 5.44 & 6.50 & 4.85 & 6.90 & 3.20 & 7.07 & 5.33 & 5.30 \\
 & 49-60 & 5.26 & 3.78 & 5.29 & 2.96 & 4.61 & 6.78 & 4.73 & 7.44 & 4.52 & 5.81 & 5.17 & 3.43 \\
 \hline
 \end{tabular*}
\end{table}

Similar to the test process in Examples \ref{exam1} and \ref{exam2}, we first determine that 
\(\xi_i(i=1,2,3)\) follow \(\mathscr{N}(5.251, 1.215)\), \(\mathscr{N}(4.982, 1.303)\), \(\mathscr{N}(5.197, 1.240)\), respectively, via uncertain hypothesis tests.
Then, by performing an uncertain homogeneity test, we find that \(H_0: e_1=e_2=\cdots =e_n~\text{and}~\sigma_1=\sigma_2=\cdots =\sigma_n\) cannot be rejected at the significance level \(\alpha=0.05\), that is to say, there is no significant difference in the unknown parameters of all populations at significance level \(\alpha\). Finally, the uncertain common test (\ref{musigma0}) is carried out after the data are combined, and the result shows that \(H_{0}^{*}\) cannot be rejected at the significance level 
\(\alpha\), that is to say the unknown parameter \(e\) is equal to a fixed constant \(e_0=5.139\), \(\sigma\) is equal to a fixed constant \(\sigma_0=1.260\).
\end{example}

\section{A real example}\label{sec5}

Here, we use partial data from the research on the population size estimation method for Calloscius erythraeus by \cite{Xu2005}. Six peeled tree banks with multiple clear upper-tooth dental marks were found in the sample plot. Suppose the upper-tooth dental marks on each peeled tree bank come from the same individual of Calloscius erythraeus. The number of individuals of Calloscius erythraeus in this sample plot is determined by analyzing the dental marks data from six peeled tree banks. The dental marks data \(z_j^{(i)} \ (j = 1, \ldots, m_i)\) from the \(i\)th peeled tree bank \(\xi_i (i = 1, 2, \ldots, 6)\) are shown in Table \ref{dental marks data from six peeled tree banks in Example 4}.

\begin{table}[h]
 \centering
 \caption{Dental marks data from six peeled tree banks.}
 \label{dental marks data from six peeled tree banks in Example 4}
 \begin{tabular*}{\hsize}{@{}@{\extracolsep{\fill}}ccccccccc@{}}
 \hline
 \(i\) &\hspace{0.5em} \(j\) &\multicolumn{7}{c}{\(z^{(i)}_j\)}\\ \hline
 1 &\hspace{0.5em} 1-6   &\hspace{0.5em} 2.8 &\hspace{0.5em} 2.8 &\hspace{0.5em} 2.9 &\hspace{0.5em} 2.9 &\hspace{0.5em} 2.9 &\hspace{0.5em} 3.0 &\hspace{0.5em}\\
 2 &\hspace{0.5em} 1-7   &\hspace{0.5em} 2.5 &\hspace{0.5em} 2.5 &\hspace{0.5em} 2.6 &\hspace{0.5em} 2.6 &\hspace{0.5em} 2.6 &\hspace{0.5em} 2.7 &\hspace{0.5em} 2.7 \hspace{0.5em}\\
 3 &\hspace{0.5em} 1-6   &\hspace{0.5em} 2.4 &\hspace{0.5em} 2.4 &\hspace{0.5em} 2.5 &\hspace{0.5em} 2.6 &\hspace{0.5em} 2.6 &\hspace{0.5em} 2.6 &\hspace{0.5em}\\
 4 &\hspace{0.5em} 1-6   &\hspace{0.5em} 2.4 &\hspace{0.5em} 2.4 &\hspace{0.5em} 2.5 &\hspace{0.5em} 2.5 &\hspace{0.5em} 2.6 &\hspace{0.5em} 2.6 &\hspace{0.5em}\\
 5 &\hspace{0.5em} 1-6   &\hspace{0.5em} 2.4 &\hspace{0.5em} 2.5 &\hspace{0.5em} 2.5 &\hspace{0.5em} 2.6 &\hspace{0.5em} 2.6 &\hspace{0.5em} 2.6 &\hspace{0.5em}\\
 6 &\hspace{0.5em} 1-7   &\hspace{0.5em} 2.4 &\hspace{0.5em} 2.4 &\hspace{0.5em} 2.5 &\hspace{0.5em} 2.5 &\hspace{0.5em} 2.6 &\hspace{0.5em} 2.6 &\hspace{0.5em} 2.6 \hspace{0.5em}\\
 \hline
 \end{tabular*}
\end{table}

The uncertain homogeneity test is used to determine whether the means and standard deviations of the dental marks on different peeled tree banks are the same. If so, we consider there to be only one individual of Calloscius erythraeus in the sample area. If not, we will further analyze these differences to determine if there may be multiple individuals of Calloscius erythraeus. 

Before performing the uncertain homogeneity test, it is necessary to determine that the dental marks data from each peeled tree bank follows a normal uncertainty distribution. Let's take the data of \(\xi_1\) as an example. Using the data \((z^{(1)}_1, z^{(1)}_2,..., z^{(1)}_{6})\) of \(\xi_1\) in Table \ref{dental marks data from six peeled tree banks in Example 4}, we can calculate that \(e_{10} =2. 883, \sigma_{10} =0. 069\) by using moment estimation. To test whether \(\xi_1 \sim \mathscr{N}(2. 883, 0. 069)\), we consider the following hypotheses:
 \begin{equation*}
\begin{aligned}
H_0^1:e_1=2. 883~\text{and}~\sigma_1=0. 069 \ \text{versus }
H_1^1:e_1\neq 2. 883~\text{or}~\sigma_1 \neq 0. 069,\quad~
\end{aligned}
\end{equation*}
 given a significance level \(\alpha\) = 0.05, we obtain
 \(
 \Phi^{-1}(\alpha/2;e_{10}, \sigma_{10})=2. 745~, ~\Phi^{-1}(1-\alpha/2;e_{10}, \sigma_{10})=3. 022,
 \)
 where 
 $
 \Phi_1^{-1}(\alpha;e_{10}, \sigma_{10})=2. 883+\frac{0. 069\sqrt{3}}{\pi}\ln_{}{(\frac{\alpha}{1-\alpha})}.
$
 Since \(\alpha \times 6 = 0. 3\),
\begin{equation*}
W_1 = \left\{
\begin{aligned}
& (z^{(1)}_1, z^{(1)}_2, \dots, z^{(1)}_{6}) : \text{there are at least } 1\text{ of indexes } p \text{'s} \\
&~~\text{with }1 \leq p \leq 6 \text{ such that }z_p < 2.745 \text{ or } z_p > 3.022
\end{aligned}
\right\}.
\end{equation*}

We can see that all the \(z_j^{(1)}\) are in \([2.745, 3.022]\). Thus \((z^{(1)}_1, z^{(1)}_2,..., z^{(1)}_{6}) \notin W_1\). \(H_0^1\) cannot be rejected. The results of the hypothesis test of \(\xi_i (i=1, 2,..., 6)\) are shown in Table \ref{Results of the hypothesis test in real Example}. Since \(\alpha \times 7 = 0.35\), the results show no singular point in the hypothesis test of each \(\xi_i\). Thus \((z^{(i)}_1, z^{(i)}_2,..., z^{(i)}_{m_i}) \notin W_i\), \(H_0^i\) cannot be rejected, each \(\xi_i\) follows \(\mathscr{N}(e_{i0}, \sigma_{i0})\).


\begin{table}[h]
\setlength{\tabcolsep}{4.15em}
\centering
\caption{Results of the hypothesis test on dental marks data.}
\label{Results of the hypothesis test in real Example}
\begin{tabular*}{\hsize}{@{}cccccc@{}}
\hline
\hspace{1em}\(z^{(i)}\) & \(e_{i0}\) & \(\sigma_{i0}\) & \(AI(z^{(i)};e_{i0}, \sigma_{i0})\) & Singular Point \\ 
\hline
\hspace{1em}\(z^{(1)}\) & 2.883 & 0.069 & [2.745, 3.022] & 0 \\ 
\hspace{1em}\(z^{(2)}\) & 2.600 & 0.076 & [2.447, 2.753] & 0 \\ 
\hspace{1em}\(z^{(3)}\) & 2.517 & 0.090 & [2.335, 2.698] & 0 \\ 
\hspace{1em}\(z^{(4)}\) & 2.500 & 0.082 & [2.335, 2.665] & 0 \\ 
\hspace{1em}\(z^{(5)}\) & 2.533 & 0.075 & [2.383, 2.684] & 0 \\ 
\hspace{1em}\(z^{(6)}\) & 2.514 & 0.083 & [2.346, 2.683] & 0 \\ 
\hline
\end{tabular*}
\end{table}

To test whether the mean \(e_i\) and the standard deviation \(\sigma_i\) of each uncertain population are equal, respectively, we consider the following hypotheses:
\[
H_0:e_1=e_2=... =e_6~\text{and}~\sigma_1=\sigma_2=... =\sigma_{6}~ \text{versus } H_1:e_1, e_2,..., e_6~\text{or}~\sigma_1, \sigma_2,..., \sigma_{6}~\text{are not all equal},
\]
it follows from Theorem \ref{c3} that the rejection region of \(H_0\) at the significance level \(\alpha=0.05\) is
 \begin{equation*}
W = \left\{
\begin{aligned}
&(z^{(i)}_1, z^{(i)}_2, \dots, z^{(i)}_{m_i}):\exists~i \neq j~(i,j\in\{1,2,...,6\}) \text{ such that at least } \alpha \text{ of indexes }p \text{'s} \\
&\text{ with } 1 \leq p \leq m_i \text{ satisfy }z^{(i)}_p < \Phi^{-1} (\alpha/2; e_{j0},\sigma_{j0})\text{ or } z^{(i)}_p > \Phi^{-1} (1-\alpha/2;e_{j0}, \sigma_{j0})
\end{aligned}
\right\},
\end{equation*}
where
$
\Phi^{-1} (\alpha; e_{j0}, \sigma_{j0})=e_{j0}+\frac{\sigma_{j0}\sqrt{3}}{\pi}\ln_{}{(\frac{\alpha}{1-\alpha})} =e_{j0}+\sigma_{j0}\Phi_0^{-1}(\alpha). 
$

In Table \ref{Results of the hypothesis test in real Example} and Fig. \ref{results of the hypothesis test in Example (f)}, we can see that more than 1 \(z_j^{(1)}\) not in \(AI(z^{(1)};e_{i0}, \sigma_{i0})(i=2,3,4,5,6)\), more than 1 \(z_j^{(2)}\) not in \(AI(z^{(2)};e_{i0}, \sigma_{i0})(i=1,3,4,5,6)\), more than 1 \(z_j^{(k)}(k=3,4,5,6)\) not in \(AI(z^{(k)};e_{i0}, \sigma_{i0})(i=1,2)\), therefore \(H_{0}\) is rejected. It can be seen that the differences in the means or standard deviations are mainly between \(\xi_1, \xi_2\) and \(\xi_k(k=3, 4, 5, 6)\) at the significance level of 0.5. 
\begin{figure}[h]
\centering
\includegraphics[width=6in]{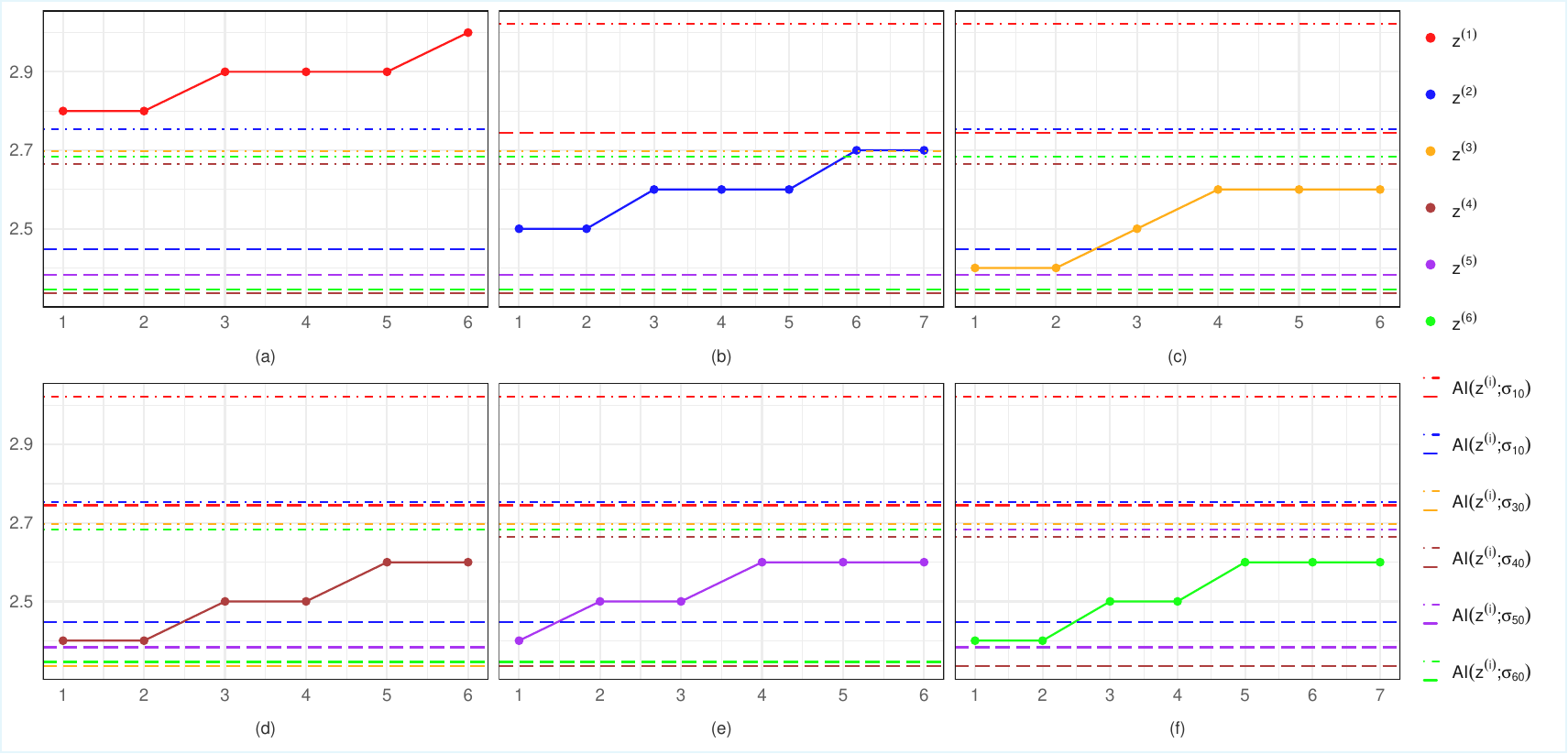}
\caption{Results of the homogeneity test on dental marks data.
}
\label{results of the hypothesis test in Example (f)}
\end{figure}
Based on the results of the uncertain homogeneity test, it can be concluded that three individuals of Callosciurus erythraeus were present in this sample plot. 

Since there is no significant difference in the unknown parameters of \(\xi_k(k=3, 4, 5, 6)\), the uncertain common test can be performed to determine whether the unknown parameter \(e, \sigma\) is equal to a fixed constant \(e_0, \sigma_0\), respectively. For this, we consider the following hypotheses:
\[
H_{0}^{*}: e = e_0 \text{ and } \sigma=\sigma_0\text{ versus } H_{1}^{*}: e \neq e_0 \text{ or } \sigma \neq \sigma_0.
\]
4 sets of data are then combined into one set, i.e.,
\begin{align*}
\{z_{1}, z_{2},..., z_{25}\}= \{&z_{1}^{(3)}, z_{2}^{(3)},..., z_{6}^{(3)}, z_{1}^{(4)}, z_{2}^{(4)},..., z_{6}^{(4)},z_{1}^{(5)}, z_{2}^{(5)},..., z_{6}^{(5)}, z_{1}^{(6)}, z_{2}^{(6)},..., z_{7}^{(6)}\},
\end{align*}
we can calculate
\(
e_0 = \frac{1}{25} \sum_{i=1}^{25} z_i=2. 516, ~
\sigma_0^2 = \frac{1}{25} \sum_{i=1}^{25} (z_i - e_0)^2=0. 083^2, 
\)
given a significance level \(\alpha=0.05\), we obtain
\(
\Phi^{-1}(\alpha/2;e_0, \sigma_0)=2. 348, \Phi^{-1}(1-\alpha/2;e_0, \sigma_0)=2. 684,
\)where 
 \(
 \Phi^{-1}(\alpha;e_0, \sigma_0)=2. 516+0. 083\frac{\sqrt{3}}{\pi}\ln_{}{(\frac{\alpha}{1-\alpha})}.
 \)
Since \(\alpha \times 25 =1. 25\), then
\begin{equation*}
W = \left\{
\begin{aligned}
&~(z_{1}, z_{2}, \dots, z_{25}) : \text{there are at least } 2 \text{ of indexes } p \text{'s}\\&\text{with } 1 \leq p \leq 25 \text{ such that } z_p < 2.348 \text{ or } z_p > 2.684
\end{aligned}
\right\}.
\end{equation*}
As shown in Table \ref{dental marks data from six peeled tree banks in Example 4}, we can see that all the \(z_j\) fall in \([2.348, 2.684]\). Thus \((z_{1}, z_{2},..., z_{25}) \notin W\), \(H_{0}^{*}\) cannot be rejected, that is to say the unknown parameter \(e\) is equal to a fixed constant \(e_0=2.516\), \(\sigma\) is equal to a fixed constant \(\sigma_0=0.083\). 
From this, we can conclude that the upper-tooth dental marks of these three Calloscius erythraeus individuals each follow a normal uncertainty distribution, specifically \( \mathscr{N}(2.883, 0.069) \), \( \mathscr{N}(2.600, 0.076) \), and \( \mathscr{N}(2.516, 0.083) \). 

Unlike the probabilistic method, which could only identify two distinct individuals (grouped as $\xi_1$ and ${\xi_2, \xi_3, \xi_4, \xi_5, \xi_6}$), the uncertain hypothesis test successfully discriminated three individuals. This reveals subtle underlying differences that might be obscured in traditional probabilistic analysis, especially with limited data.

\section{Conclusion}\label{sec6}
This paper establishes a theoretical framework for hypothesis testing across multiple uncertain populations, introducing the uncertain family-wise error rate (UFWER) as a foundation for extending the methodology of Ye and Liu from a single population to multi-population settings.
The framework is built upon two procedures: the uncertain homogeneity test and the uncertain common test. 
While the uncertain homogeneity test assesses whether unknown parameters are equivalent across populations, the uncertain common test extends this analysis, focusing on whether these unknown parameters are equal to a fixed constant.
Uncertain homogeneity tests and uncertain common tests for normal uncertain populations are divided into three cases depending on whether the means and standard deviations of the populations are known. An uncertain homogeneity test can be effectively achieved by constructing a new rejection region. 
An uncertain common test is conducted by performing a single-population hypothesis test under uncertainty after appropriately adjusting and combining the data. Numerical examples demonstrating both the uncertain homogeneity test and the uncertain common test are provided.
Finally, the paper offers a real example of population size estimation for Calloscius erythraeus, demonstrating not only the efficacy of the proposed method but also its favorable sensitivity and practicality in small-sample scenarios, thereby providing a valuable supplement to traditional probabilistic methods.
For future research, there is a clear avenue to explore the application of this method to non-normal uncertain populations and to develop more sophisticated statistical models.

\section*{Acknowledgments} The work was supported by the National Natural Science Foundation of China (12561047), and the 2025 Central Guidance for Local Science and Technology Development Fund (ZYYD2025ZY20).

\section*{Conflict of Interest}
The authors declare that they have no conflict of interest.
\bibliography{sn-bibliography}

\end{document}